\documentclass [11pt] {article}

 \usepackage{fullpage}

\usepackage{graphics}
\usepackage[dvips]{epsfig}

\usepackage{amsmath}
\usepackage{amssymb}
\usepackage{amsfonts}
\usepackage{graphicx}

\usepackage{cite}

\usepackage{algorithmic}
\usepackage[linesnumbered,ruled]{algorithm2e}

\usepackage{array}
\usepackage{color}

\usepackage{multirow}
\usepackage{rotating}

\begin{document}

\newtheorem{theorem}{Theorem}[section]
\newtheorem{lemma}{Lemma}[section]
\newtheorem{corollary}{Corollary}[section]
\newtheorem{claim}{Claim}[section]
\newtheorem{proposition}{Proposition}[section]
\newtheorem{definition}{Definition}[section]
\newtheorem{fact}{Fact}[section]
\newtheorem{example}{Example}[section]

\newcommand{\cA}{{\cal A}}
\newcommand{\cB}{{\cal B}}
\newcommand{\cC}{{\cal C}}
\newcommand{\cG}{{\cal G}}
\newcommand{\cN}{{\cal N}}
\newcommand{\cU}{{\cal U}}
\newcommand{\cT}{{\cal T}}
\newcommand{\cS}{{\cal S}}
\newcommand{\cL}{{\cal L}}
\newcommand{\cV}{{\cal V}}
\newcommand{\loc}{{\cal LOCAL}}

\newcommand{\ai}{\alpha_i}
\newcommand{\bi}{\beta_i}
\newcommand{\gi}{\gamma_i}
\newcommand{\di}{\delta_i}

\newcommand{\oai}{\overline{\alpha}_i}
\newcommand{\obi}{\overline{\beta}_i}
\newcommand{\ogi}{\overline{\gamma}_i}
\newcommand{\odi}{\overline{\delta}_i}

\newcommand{\qed}{\hfill $\square$ \smallbreak}
\newenvironment{proof}{\noindent{\bf Proof:}}{\qed}

\newcommand{\procend}{\hfill $\diamond$\medskip}

\newcommand{\oddRepair}{{\tt Odd\-Repair}}
\newcommand{\Deactivate}{{\tt Deactivate}}
\newcommand{\evenARepair}{{\tt Even\-Al\-most\-Re\-pair}}
\newcommand{\ringThree}{{\tt Ring\-Three\-Co\-lo\-ring}}
\newcommand{\ringLearning}{{\tt Ring\-Lear\-ning}}
\newcommand{\Elect}{{\tt Elect}}
%%%%%% NEW-MACRO
%\def\lalto{\left \lceil}
%\def\ralto{\right \rceil}
%\def\lbasso{\left \lfloor}
%\def\rbasso{\right \rfloor}
%\def\D{{\Delta}}
%\def\qed{\hfill$\Box$}

%\baselineskip    0.175in
%\parskip         0.0in
%%%\parindent       0.3in

%\baselineskip    0.17in
%\parskip         0.1in
%\parindent       0.1in

\def\thefootnote{\fnsymbol{footnote}}

\title{{\bf Using Time to Break Symmetry:\\
 Universal Deterministic Anonymous Rendezvous}}
%\title{{\bf Exact trade-offs between time and advice\\ for topology recognition}}

\author{Andrzej Pelc\footnotemark[1]
\and
Ram Narayan Yadav\footnotemark[2]
}

\footnotetext[1]{
 D\'epartement d'informatique, Universit\'e du Qu\'ebec en Outaouais, Gatineau,
Qu\'ebec J8X 3X7, Canada. {\tt pelc@uqo.ca}. Partially supported by NSERC discovery grant 8136--2013
and by the Research Chair in Distributed Computing at the
Universit\'e du Qu\'ebec en Outaouais.}

\footnotetext[2]{D\'epartement d'informatique, Universit\'e du Qu\'ebec en Outaouais, Gatineau,
Qu\'ebec J8X 3X7, Canada. {\tt narayanram.1988@gmail.com}}

\maketitle

\thispagestyle{empty}

\begin{abstract}
Two anonymous mobile agents navigate synchronously in an anonymous graph and have to meet at a node, using a deterministic algorithm. This is a symmetry breaking task called rendezvous, equivalent to the fundamental task of leader election between the agents. When is this feasible in a completely anonymous environment? It is known that agents can always meet if their initial positions are nonsymmetric, and that if they are symmetric and agents start simultaneously then rendezvous is impossible. What happens for symmetric initial positions with non-simultaneous start? Can symmetry between the agents be broken by the delay between their starting times?

In order to answer these questions, we consider {\em space-time initial configurations} (abbreviated by STIC). A STIC is formalized as $[(u,v),\delta]$, where $u$ and $v$ are initial nodes of the agents in some graph and $\delta$ is a non-negative integer that represents the difference between their starting times. A STIC is {\em feasible} if there exists a deterministic algorithm, even dedicated to this particular STIC, which accomplishes rendezvous for  it. Our main result is a characterization of all feasible STICs and the design of a universal deterministic algorithm that accomplishes rendezvous for all of them without {\em any } a priori knowledge of the agents. Thus, as far as feasibility is concerned, we completely solve the problem of symmetry breaking between two anonymous agents in anonymous graphs.
%Hence our algorithm breaks symmetry between the agents whenever it is possible, without  knowing anything about the graph,  the positions of the agents in it, or the delay between them. 
Moreover, we show that such a universal algorithm cannot work for all feasible STICs in time polynomial in the initial distance between the agents.

\vspace*{0.5cm}

\noindent
{\bf keywords:} anonymous graph, anonymous agent, rendezvous, symmetry breaking, universal algorithm

\vspace*{4cm}
\end{abstract}

\setcounter{page}{0}
\pagebreak

%%%%%%%%%%%%%%%%%%%%%%%%%%%%%%%%%%%%
\section{Introduction}
%%%%%%%%%%%%%%%%%%%%%%%%%%%%%%%%%%%%
Two anonymous mobile agents start from two nodes of an anonymous graph and have to meet at a node, using the same deterministic algorithm.  The agents navigate in synchronous rounds but may start with arbitrary delay, chosen by the adversary. In applications, agents may be mobile robots moving in corridors of a contaminated mine and taking samples of the ground, or they may be software agents navigating in a computer network and consulting a distributed database located in its nodes.  The task of meeting, also called rendezvous, is an extensively studied symmetry breaking task, equivalent to the fundamental task of leader election between the agents. When is it feasible in a completely anonymous environment? It is known that agents can always meet if their initial positions are nonsymmetric\footnote{Nodes are symmetric if they have the same views: see the precise definition in Section 2.}, and that if they are symmetric and agents start simultaneously then rendezvous is impossible. What happens for symmetric initial positions with non-simultaneous start? Can symmetry between the agents be broken by the delay between their starting times?

\noindent
{\bf The model and the problem.}
We consider simple finite undirected connected graphs $G=(V,E)$. The number of nodes of the graph is called its size. Nodes of the graph are unlabeled, but ports at each node of degree $d$ are labeled by integers $0,1,\dots, d-1$. There is no any coherence between port numbers at two extremities of an edge.

Two mobile agents start at different nodes of a graph and have to meet at the same node. Agents navigate in the graph in synchronous rounds.
In each round, an agent can either stay at the current node or move to an adjacent node by a chosen port. When an agent arrives at a node, it sees its degree and the port number by which it enters the node. Agents are identical (they have no labels) and execute the same deterministic algorithm. They have an unbounded memory: from the computational point of view they are modeled as Turing machines. They start in possibly different rounds, chosen by the adversary. The difference between these starting times is called the {\em delay} which can be any non-negative integer. Agents are equipped with synchronized clocks ticking once per round. An agent appears at its initial  node in the starting round and its clock starts at this time. Agents do not have any a priori knowledge: they don't know anything about the graph in which they navigate, they don't know their initial positions or the delay between their starting rounds. They don't have access to any global clock.
A {\em rendezvous} between agents occurs when they are at the same node in the same round. The time of a rendezvous algorithm is the worst-case number of rounds between the appearance of the later agent and the meeting. Agents can cross each other in an edge going in opposite directions but they do not notice it.

The rendezvous task is a symmetry-breaking problem. In fact, it is equivalent to the most fundamental symmetry-breaking problem, that of leader election \cite{Ly}. Applied to anonymous mobile agents, leader election calls for one of them to become the leader and the other to become non-leader.
To see the equivalence between rendezvous and leader election, suppose first that roles leader and non-leader are assigned to the agents.
Then the non-leader can wait at its initial node and the leader explores the graph and finds it. This algorithm, called ``waiting for Mommy''
 reduces rendezvous to exploration when the leader is elected. Conversely, suppose that the agents have met. Then they can compare their trajectories coded as sequences of encountered port numbers. Since agents were at different nodes at the beginning and they succeeded to meet, there must be some node to which the agents entered by different ports. Consider the last such node $v$ before meeting (in particular it may be the meeting node). The agent that entered node $v$ by the port with larger number can be elected as leader.  This equivalence between rendezvous and leader election justifies the importance of rendezvous and shows its fundamental symmetry-breaking character. 

Since rendezvous is a symmetry-breaking problem, agents have to do something differently in order to meet. Since they are identical and execute the same deterministic algorithm, this difference can come only from two sources: either  a space difference, i.e, the structure of the graph from the point of view of each agent looks different and entices them to make different decisions, or a time difference, i.e., agents make possibly the same decisions but at different times, which enables them to meet. The first type of difference can be, e.g., a different degree of the initial node which can be adopted as the label of the agent and facilitate rendezvous. The second difference can be exemplified in the two-node graph. If identical agents start in this graph with delay 3, executing the algorithm ``move at each round'', then they will meet 3 rounds after the start of the earlier agent.

It follows from \cite{CKP} that if agents start from nonsymmetric positions with arbitrary delay and know  some upper bound on the size of the graph then they can meet. On the other hand, if they start {\em simultaneously} from symmetric positions then meeting is impossible because  in each round they will be at different symmetric positions, regardless of the algorithm. This leaves the third possibility: what happens if agents start from symmetric positions with a positive delay? Another natural question is if some knowledge about the graph, about the initial positions or about the delay is necessary for meeting.

In order to answer these questions, we consider {\em space-time initial configurations} (abbreviated by STIC). A STIC is formalized as $[(u,v),\delta]$, where $u$ and $v$ are initial nodes of the agents in some graph and $\delta$ is a non-negative integer that represents the difference between their starting rounds. A STIC is {\em feasible} if there exists a deterministic algorithm, even dedicated to this particular STIC, which accomplishes rendezvous for  it. Now the questions stated above can be reformulated as follows.

 $\bullet$ Which STICs are feasible? Does there exist a universal deterministic algorithm that accomplishes rendezvous for all feasible STICs without any a priori knowledge?

\noindent
{\bf Our contribution.}
Our main result is a characterization of all feasible STICs and the design of a universal deterministic algorithm that accomplishes rendezvous for all of them without {\em any } a priori knowledge of the agents. Thus, as far as feasibility is concerned, we completely solve the problem of symmetry breaking between two anonymous agents in anonymous graphs. Moreover, we show that such a universal algorithm cannot work for all feasible STICs in time polynomial in the initial distance between the agents.

\noindent
{\bf Related work.}
The rendezvous problem was extensively studied in the literature, both in the deterministic and in the randomized scenario.
An excellent survey of  randomized rendezvous in various models  can be found in
\cite{alpern02b}, cf. also  \cite{alpern95a,alpern02a,anderson90,baston98}. 
Deterministic rendezvous in networks was surveyed in \cite{Pe}.
Several authors
considered the geometric scenario (rendezvous in an interval of the real line, see, e.g.,  \cite{baston98,baston01,gal99},
or in the plane, see, e.g., \cite{anderson98a,anderson98b}).
The extension of the rendezvous problem to several agents is usually called gathering, and was studied, e.g., 
in \cite{lim96,thomas92}. Gathering many labeled agents in the presence of Byzantine agents was studied in \cite{BDD,DPP}. 
The gathering problem was also studied in the context of oblivious robot systems in the plane, cf.
\cite{CP05,FPSW}, and fault tolerant gathering of robots in the plane was studied, e.g., in \cite{AP06,CP08}. 

For the deterministic setting, attention was usually concentrated on the feasibility of rendezvous, and on the time required to achieve this task, when feasible. For instance, deterministic rendezvous with agents equipped with tokens used to mark nodes was considered, e.g., in~\cite{KKSS}. Deterministic rendezvous of two agents with unique labels was discussed in \cite{DFKP,KM,TSZ}.
These papers considered the time of rendezvous in arbitrary
graphs. In \cite{CKP,DFKP} the authors showed a rendezvous algorithm polynomial in the size of the graph, in the length of the shorter
label and in the delay between the starting times of the agents. In \cite{KM,TSZ} rendezvous time was polynomial in the first two of these parameters and independent of the delay.
 In \cite{CKP,FP} the optimization criterion for rendezvous was the memory size of the agents:
it was studied in \cite{FP} for trees and in  \cite{CKP} for general graphs.
Memory needed for randomized rendezvous in the ring was discussed, e.g., in~\cite{KKPM08}. 

Apart from the synchronous model used in this paper, several authors considered asynchronous rendezvous in the plane \cite{CFPS,FPSW} and in networks 
\cite{BCGIL,CLP,DGKKP,DPV,GP}.
In the latter scenario the agent chooses the edge which it decides to traverse but the adversary controls the speed of the agent. Under this assumption rendezvous
in a node cannot be guaranteed even in very simple graphs and hence the rendezvous requirement is relaxed and agents are permitted to meet inside an edge. 

Computational tasks in anonymous networks were studied in the literature, starting with the seminal paper \cite{A},
followed, e.g., by \cite{ASW, BV,KKV}. While the considered tasks, such as leader election or computing Boolean functions differ
from rendezvous studied in the present paper, the main concern is usually symmetry breaking, similarly as in our case.

Deterministic rendezvous of anonymous agents in arbitrary anonymous graphs was previously studied in \cite{CKP,DP1,GP}. Papers \cite{CKP,DP1} were concerned with the synchronous scenario. The main result of \cite{CKP} was a rendezvous algorithm working for all nonsymmetric initial positions using memory logarithmic in the size of the graph. \cite{DP1} was concerned with gathering multiple anonymous agents and characterized initial positions that allow gathering with all starting times. The authors of \cite{GP} characterized initial positions that allow asynchronous rendezvous. None of these papers considered the issue  of breaking symmetry using time.

\section{Preliminaries}

Let $G$ be any graph and $v$ a node in this graph. The following notion is crucial for our considerations. The {\em view} from $v$ in $G$, denoted $\cV(v,G)$, is the tree of all paths in $G$, starting from node $v$ and coded as sequences of port numbers, where the rooted tree structure is defined by the prefix relation of sequences. This definition is equivalent  to that from \cite{YK3}. Nodes $u$ and $v$ in $G$ are called {\em symmetric}, if $\cV(u,G)=\cV(v,G)$.

It follows from \cite{CKP} that rendezvous of anonymous agents is possible for any STIC $[(u,v),\delta]$, where $u$ and $v$ are nonsymmetric nodes in a graph of size at most $n$ and $\delta$ is any non-negative integer, in time polynomial in $n$.

 Let $G$ be any graph.
For any node $v$ of degree $d(v)$ of $G$ and any integer $0\leq p \leq d(v)-1$, denote by $succ(v,p)$ the neighbor $w$ of $v$, such that 
the edge $\{v,w\}$ corresponds to the port $p$ at $v$.  Let $(a_1,a_2,...,a_k)$ be a sequence of integers.
An application of this sequence in the graph $G$ at node $u$ is the sequence of nodes $(u_0,u_1,...,u_{k+1})$ obtained as follows: $u_0 = u$, $u_1=succ(u_0,0)$; for any $1 \leq i \leq k$, $u_{i+1} = succ(u_i,(p+a_i))$ mod $d(u_i)$, where $p$ is the port number at node $u_i$ corresponding to the edge $(u_{i-1},u_i)$. We will use the notion of a Universal Exploration Sequence (UXS) \cite{Ko}. A sequence $Y(n)=(a_1,a_2,...,a_M)$ of integers, whose application $R(u)=(u_0,u_1,...,u_{M+1})$ in any graph of size $n$ at any node $u$ of this graph contains all the nodes of the graph is called a UXS
for the class of graphs of size $n$. It follows from \cite{Re} that there exists a polynomial length UXS
for the class of graphs of size $n$.

%Suppose that a sequence $(a_1,a_2,...,a_k)$ is applied at a node $u$ of $G$ results a sequence of visited nodes $(u_0,u_1,...,u_{k+1})$, the corresponding path is called forward path and represented as FwdPath$(u_0,u_1,...u_{k+1})$. The procedure FwdPath$(u_0,u_1,...u_k)$ traverse $k$ consecutive nodes along the application of $Y(n)$ at node $u_0$ in the graph $G$. Procedure RevsPath$(u_k,u_{k-1},...,u_0)$ performs the backward walk along the application of $Y(n)$ starting at node $u_k$ and using the formula $succ(u,(p-a_i))$ mod $d(u)$ to determine the previous node in reverse path. To start the backward walk from the node $u_k$, the port by which node $u_k$ was entered is used first and elements $a_{k-1},a_{k-1},...,a_1$ of $Y(n)$ are used in the order of decreasing indices. 
 
Consider a sequence of nodes $\pi=(u_1,...,u_{k+1})$ forming a path in a graph $G$. Suppose that $succ(u_i,p_i)=u_{i+1}$
and that $succ(u_{i+1},q_i)=u_{i}$, for $i=1,\dots,k$. Hence $\pi$ starts at $u_1$ and corresponds to the sequence of outgoing port numbers $(p_1,\dots, p_{k})$.
We define the reverse path $\overline{\pi}$ as the path $(u_{k+1}, u_k,\dots, u_1)$ starting at $u_{k+1}$ and corresponding to the 
sequence of outgoing port numbers $(q_k,\dots,q_1)$.

For any node $x$ of a graph $G=(V,E)$ and any sequence $\alpha = (p_1,p_2,...,p_s)$ of port numbers, define 
$\alpha(x)$ as the node $y$, such that the path from node $x$ following the consecutive outgoing port numbers $p_1,p_2,...,p_s$ ends at node $y$.

\section{The universal algorithm}

In this section, we characterize all feasible STICs and design a universal deterministic algorithm that accomplishes rendezvous for all of them without any  a priori knowledge of the agents. 
In our characterization we will use the following notion.

%Consider a graph $G(V,E)$ and any symmetric pair of nodes $x,y \in V(G)$. 
%he following notion is crucial for our considerations.

\begin{definition}%\label{shrink}
For any graph $G=(V,E)$ and any symmetric pair of nodes $u,v \in V$, $Shrink(u,v)$ is the smallest distance between $\alpha(u)$ and $\alpha(v)$,  over all possible sequences $\alpha$ of port numbers.
\end{definition}

Hence, intuitively, for any symmetric pair $u,v$ of nodes, $Shrink(u,v)$ is the smallest distance at which it is possible to get from nodes $u$ and $v$, by applying the same sequence of port numbers. For example, in an oriented torus, any pair of nodes is symmetric, and
$Shrink(u,v)$ is equal to the distance between $u$ and $v$, for any $u$ and $v$. By contrast, in a symmetric tree that is composed of a central edge with port-preserving isomorphic trees attached to both of its ends,  $Shrink(u,v)$ for any symmetric pair $(u,v)$ of nodes is always 1, although the distance between $u$ and $v$ can be arbitrarily large (i.e., in this case $Shrink$ can really shrink the initial distance).

The following result gives a necessary condition for feasibility of rendezvous starting from symmetric initial positions.

\begin{lemma}\label{infeasible}
For any symmetric pair $u$ and $v$ of nodes, a STIC $[(u,v),\delta]$ with delay $\delta <Shrink(u,v)$ is not feasible.
\end{lemma}
\begin{proof}
We prove the lemma by contradiction. Suppose that rendezvous is feasible and $\delta <Shrink(u,v)$. 
Consider any deterministic rendezvous algorithm.
Let $t$ be the earliest time of rendezvous, counted from the start of the later agent. Since initial positions of the agents are symmetric, in each round the agents follow the edges of paths having the same outgoing ports and the same incoming ports. Since the earlier agent starts $\delta$ rounds ahead of the later agent,
the path traversed by the later agent until rendezvous follows the same sequence of port numbers as the path of the earlier agent until round $t-\delta$ from the start of the later agent. Call this round $t'$.

Let the distance between the later agent in round $t$ and the earlier agent in round $t'$ be $d'$. For rendezvous, the earlier agent should cover distance $d'$ during the  last $\delta$ rounds after time $t'$. By the definition of $Shrink(u,v)$ we have
$d' \geq Shrink(u,v)$. However, by assumption,  $\delta < Shrink(u,v)$, which gives a contradiction.
\end{proof}

\subsection{Rendezvous with known parameters from symmetric positions}

We first describe a procedure that solves the rendezvous problem for symmetric initial positions $(u,v)$ of the agents and a delay $\delta \geq Shrink(u,v)$, assuming that the size $n$ of the graph, the value $d$ of $Shrink(u,v)$ and  the delay $\delta$ are known. We will then use this procedure to solve the rendezvous problem in our scenario when the agents are ignorant of these parameters.

The procedure uses an application $R(u)$ of a UXS for the family of graphs of size $n$ at the starting position $u$ of the agent. It also uses procedure Explore$(u,d,\delta)$, where $d$ and $\delta$ are positive integers and $d \leq \delta$. In the procedure Explore $(u,d,\delta)$, the agent explores all possible paths of length $d$ starting at node $u$, each time backtracking along the reverse path, and waits $\delta-d$ rounds at node $u$. The idea of Procedure $SymmRV(n,d,\delta)$ is to follow $R(u)$, executing Explore$(u_i,d,\delta)$ at each node $u_i$, $1 \leq i \leq M$, of $R(u)$, and then to backtrack to $u$ along path $\overline{R(u)}$.

\begin{algorithm}
%\KwIn{Demanding clusters set $\mathcal{S}_i$,}
%SymmRV$(n,d,\delta)$ \\
%\{node $u$ and $v$ are the starting symmetric positions of the agents\}\\
%\KwOut{$i$'s membership}
\Begin
{
Let $Y(n)=(a_1,a_2,...,a_M)$ be a UXS for the class of graphs of size $n$ \\
$u_0 = u$ \\
Explore$(u_0, d,\delta)$\\
$u_1= succ(u_0,0)$\\
Explore$(u_1, d,\delta)$\\
\For{$i=1$ to $M$}
{$q:=$ the port number by which the agent enters node $u_i$\\
$u_{i+1}=succ(u_i,(q+a_i)$ mod $d(u_i))$ \\
Explore$(u_{i+1}, d,\delta)$}
Go back to $u_0$ using path $\overline{(u_0,u_1,\dots, u_{M+1})}$  }
\caption{Procedure $SymmRV(n,d,\delta)$}
\label{basic}
\end{algorithm} 
 
 \begin{algorithm}
%\KwIn{}
%Explore$(u,d,\delta)$ \\
%\KwOut{}
\Begin
{
\For{all possible paths $\pi$ of length $d$ starting at node $u$, in lexicographic order of corresponding port sequences}
{
Traverse path $\pi$\\
Traverse path $\overline{\pi}$\\
Wait $(\delta-d)$ rounds

}
 }
\caption{Procedure $Explore(u,d,\delta)$}
\label{explore}
\end{algorithm}

\begin{lemma}\label{lemma-symm}
Consider any STIC $[(u,v),\delta]$, such that $u$ and $v$ are symmetric nodes of a graph of size $n$ and 
$\delta \geq Shrink(u,v)$. A pair of agents starting from this STIC that executes
Procedure $SymmRV(n,d,\delta)$ achieves rendezvous.
\end{lemma}

\begin{proof}
Consider agents that are initially located at symmetric positions $u$ and $v$ in a graph of size $n$ and start with delay $\delta \geq Shrink(u,v)$. Without loss of generality, let the agent starting at $v$ be the later agent.
%For any sequence $\alpha=(p_1,\dots,p_{t})$ of port numbers, let $\alpha(u)$ be the application of $\alpha$ at node $u$ and let
%$\alpha(v)$ be the application of $\alpha$ at node $v$. 
Let $\beta$ be any sequence of port numbers such that $u'=\beta(u)$
is at distance $d=Shrink(u,v)$ from $v'=\beta(v)$.
Let $u_j$ be the first node in the application of UXS $Y(n)$ at $u$, such that $u'=u_j$. 
Consider the node $v_j$ in the application of UXS $Y(n)$ at $v$.
By symmetry of nodes $u$ and $v$ and by the definition of $u'$ and $v'$, we have that $v'=v_j$.

Consider the first execution $E_1$ of Procedure $Explore(u_j,d,\delta)$ during the execution of  Procedure $SymmRV(n,d,\delta)$ by the earlier agent, and the first execution $E_2$ of Procedure $Explore(v_j,d,\delta)$ during the execution of  Procedure $SymmRV(n,d,\delta)$ by the later agent. By symmetry of $u$ and $v$, the execution of $E_1$ starts $\delta$ rounds ahead of the execution $E_2$.

Let $(\pi_1,\pi_2,\dots ,\pi _s)$ be the sequence of paths in lexicographic order of corresponding port sequences, traversed in the execution of Procedure $Explore(u_j,d,\delta)$.
Suppose that the first path of length $d$ which leads from $u'$ to $v'$ is $\pi_i$. Let $t$ be the round number counted from the start of the later agent when it starts execution $E_2$. Until the end of this proof we will count all round numbers from the start of the later agent.
Each execution of the {\bf for} loop in Procedure $Explore$ lasts $2d+(\delta-d)=d+\delta$ rounds. Let $x=t+(d+\delta) (i-1)$ be the round when the later agent ends the $(i-1)$th execution of the {\bf for} loop in the execution $E_2$. 
Hence the later agent waits at node $v'$ during the time interval $[x-(\delta - d),x]$. Now the earlier agent ends 
the $(i-1)$th execution of the {\bf for} loop in the execution $E_1$ in round $x-\delta$, and finishes traversing path $\pi _i$
in round $x-\delta+d$. Thus it gets to node $v'$ at the beginning of the waiting period of the later agent, and rendezvous is accomplished. 
\end{proof}

\begin{lemma}\label{time-symm}
The maximum time of execution of the Procedure $SymmRV(n,d,\delta)$ is $T(n,d,\delta)=[(d+\delta) \cdot (n-1)^d] \cdot (M+2) +2 \cdot (M+1)$, where $M$ is the length of UXS for the class of graphs of size $n$.
\end{lemma}
\begin{proof}
Let $Y(n)=(a_1,a_2,...,a_M)$ be a UXS for the class of graphs of size $n$.
The call to Procedure $Explore(u_i,d,\delta)$ for each node $u_i$, $0 \leq i \leq M+1$, makes the agent traverse all possible paths of length $d$ starting at node $u_i$, each time backtracking along the reverse path and waiting $\delta-d$ rounds at node $u_i$. Since the size of the graph is $n$, the number of possible paths of length $d$ starting at node $u_i$ can be at most $(n-1)^d$. Hence, the Procedure $Explore(u_i,d,\delta)$ requires at most $(2d+\delta-d) \cdot (n-1)^d$ rounds at each node $u_i$, $0 \leq i \leq M+1$. 

So, the total time of traversing the graph along the path $R(u)=(u_0,u_1,...,u_{M+1})$ and executing the Procedure $Explore(u_i,d,\delta)$ at each node $u_i$, for $0 \leq i \leq M+1$, is bounded by  $[(d+\delta) \cdot (n-1)^d] \cdot (M+2) + (M+1)$ rounds. Aferwards, the agent goes back to its initial position $u_0$ using the reverse path $\overline{(u_0,u_1,\dots, u_{M+1})}$, in another $M+1$ rounds. Hence, the total execution time of the Procedure $SymmRV(n,d,\delta)$ is bounded by $T(n,d,\delta)=[(d+\delta) \cdot (n-1)^d] \cdot (M+2) + (M+1) + (M+1) = [(d+\delta) \cdot (n-1)^d] \cdot (M+2) +2 \cdot (M+1)$. 
\end{proof}

Let  $AsymmRV(n)$ be the procedure from \cite{CKP} that solves the rendezvous problem when the agents start with arbitrary delay from arbitrary nonsymmetric initial positions in any graph of size $n$. The following proposition is implied by the results of \cite{CKP}. 
\begin{proposition}\label{lemma-asymm}
For any STIC $[(u,v),\delta]$, such that $u$ and $v$ are nonsymmetric nodes in a graph of size $n$ and $\delta \geq 0$, a pair of anonymous agents starting from this STIC and executing Procedure $AsymmRV(n)$ achieves rendezvous in time at most $P(n)$, where $P(n)$ is polynomial in $n$.
\end{proposition}

\subsection{Universal anonymous rendezvous}

Now, we describe the Algorithm $UniversalRV$ that solves the rendezvous problem for any STIC $[(u,v),\delta]$ such that either $(u,v)$ are nonsymmetric nodes and $\delta$ is any non-negative integer, or $(u,v)$ are symmetric nodes and $\delta \geq Shrink(u,v)$. Algorithm $UniversalRV$ does not use any a priori knowledge whatsoever: it has no information about the graph, its size or about the STIC.

Let $\mathcal{N}$ be the set of positive integers.
We will use the function $f:\mathcal{N}\times \mathcal{N} \longrightarrow \mathcal{N}$ given by the formula $f(x,y)=x+(x+y-1)(x+y-2)/2$. 
The function $f$ is a bijection from $\mathcal{N}\times \mathcal{N}$ to  $\mathcal{N}$. Hence the function $g: \mathcal{N}\times \mathcal{N}\times \mathcal{N} \longrightarrow \mathcal{N}$ given by the formula $g(x,y,z)=f(f(x,y),z)$ is a bijection from $ \mathcal{N}\times \mathcal{N}\times \mathcal{N}$ to $\mathcal{N}$.

At a high level, Algorithm $UniversalRV$ is executed in phases $P=1,2,\dots$ and interrupted when rendezvous is achieved. For any phase $P$, let  $(n,d,\delta)=g^{-1}(P)$. In the phase corresponding to the triple $(n,d,\delta)$, the algorithm ``assumes'' that $n$ is the size of the graph, $d=Shrink(u,v)$ if the initial positions are symmetric and $\delta$ is the delay between the starting times of the agents. The phase is executed if $d<n$ because $Shrink(u,v)$ is a distance between some nodes of a graph, and hence must be smaller than its size. First Procedure $AsymmRV(n)$ is executed in the hope that the intial positions are nonsymmetric, in which case the agents should meet executing it. If this does not happen by the time prescribed by Proposition \ref{lemma-asymm}, the agent backtracks to its initial position and waits so that the rest of the phase be executed
from the initial position, with the same delay as at the start of phase 1. Then, if $\delta \geq d$,  Procedure  $SymmRV(n,d,\delta)$ is executed in the hope that 
the initial positions are symmetric. If this also fails, it means that the assumptions of the current phase were not true, and the agent starts the next phase at its initial position, after waiting sufficient time to ensure that the new phase is started with the original delay.
We will prove that rendezvous occurs at the latest during the phase corresponding to the triple $(n,d,\delta)$,
which satisfies the following conditions:  $n$ is the actual size of the graph, $\delta$ is the actual delay and 
either the initial positions $u$ and $v$ are nonsymmetric  or they are  symmetric, 
$d=Shrink(u,v)$ and $\delta \geq  Shrink(u,v)$.

%The integer $n$ from the triple $(n,d,\delta)$ is used as the input parameter for the Procedures $AsymmRV(n)$ and all the three integers from this triple are used as parameters of Procedure  $SymmRV(n,d,\delta)$ in phase $P$. For triple $(n,d,\delta)$ corresponding to any phase $P$, the Procedure $AsymmRV(n,\delta)$ uses parameters $n, \delta$ and ignores parameter $d$. The Procedure $AsymmRV(n,\delta)$ executes for $P(n,\delta)$ rounds in the phase $P$ since the execution time of $AsymmRV(n,\delta)$ by each agent is at most $P(n,\delta)$. To preserve STIC $[(u,v),\delta]$, agents backtrack to their initial position $u$ and $v$ along the paths $\overline{\pi}_u$ and $\overline{\pi}_v$, respectively, where $\overline{\pi}_u$ and $\overline{\pi}_v$ are the paths traversal during the execution of the Procedure $AsymmRV(n,\delta)$ by the agents from the initial position $u$ and $v$, respectively. To synchronize agents to terminate the procedure $AsymmRV(n,\delta)$ at the same time, agent waits until $2 P(n,\delta)$ rounds from its start of the Procedure $AsymmRV(n,\delta)$ execution. Moreover, the Procedure $SymmRV(n,d,\delta)$ executes only if $\delta \geq d$ and uses all three parameters $(n,d,\delta)$ of corresponding phase $P$. To synchronize, each agent waits until $T(n,d,\delta)$ rounds from the start of the Procedure $SymmRV(n,d,\delta)$ execution.
%The Procedure $UniversalRV$ executes forever until rendezvous is achieved.
Algorithm 3  gives the pseudocode of our universal algorithm. It is interrupted as soon as the rendezvous is achieved (which can occur in the middle of a phase).
 \begin{algorithm}
%\KwIn{}
%Explore$(u,d,\delta)$ \\
%\KwOut{}
\Begin
{
$P:=1$ \\
Repeat forever\\
\Begin{
 $(n,d,\delta):=g^{-1}(P)$\\
 \If{$d<n$}{
 Execute Procedure $AsymmRV(n)$ for $P(n)+\delta$ rounds\\
	{
 $\pi:=$ the path traversed during the execution of $AsymmRV(n)$ from the initial position $u$\\
 Backtrack to the initial position $u$ along the path $\overline{\pi}$ \\
 Wait until $2 (P(n)+\delta)$ rounds from the start of Procedure $AsymmRV(n)$  
	} \\
\If{$\delta \geq d$}
{
Execute Procedure $SymmRV(n,d,\delta)$\\
Wait until $T(n,d,\delta)$ rounds from the start of Procedure $SymmRV(n,d,\delta)$ 
}
}
$P:=P+1$
 }}
\caption{ $UniversalRV$}
\label{universal}
\end{algorithm}

\begin{theorem}\label{main}
Consider any STIC $[(u,v),\delta]$, such that either $(u,v)$ are nonsymmetric nodes and $\delta$ is any non-negative integer or $(u,v)$ are symmetric nodes and $\delta \geq Shrink(u,v)$. A pair of agents starting from this STIC that executes Algorithm $UniversalRV$ achieves rendezvous. 
\end{theorem}
\begin{proof}
Consider any STIC $[(u,v),\delta]$ in a graph $G$ of size $n$. First consider the case when $u$ and $v$ are nonsymmetric. Then the agents meet at the latest in the first phase $P$ such that $g^{-1}(P)=(n,d,\delta)$ by Proposition
\ref{lemma-asymm}. Notice that, while Proposition \ref{lemma-asymm} guarantees the meeting of the agents starting from nonsymmetric positions in time at most $P(n)$ (counted, as usual, from the start of the later agent), in our algorithm each agent executes Procedure $AsymmRV(n)$ for $P(n)+\delta$ rounds because the agent does not know whether it is earlier or later and hence  it must make this precaution, in order to ensure sufficient time for meeting in the case when it is the earlier agent. (In \cite{CKP} it was {\em assumed} that the initial positions are nonsymmetric, hence agents simply executed the procedure until rendezvous).  

Hence we may suppose that nodes $u$ and $v$ are symmetric and $\delta \geq Shrink(u,v)$.
Consider the first phase $P'$ such that $g^{-1}(P')=(n,d',\delta)$, where $Shrink(u,v)=d'$. 
For every phase $P''<P'$, each of the agents uses the same number of rounds to execute it. This is due to the waiting times after each execution of procedures $AsymmRV$ and $SymmRV$. Moreover, each of the agents uses the same number of rounds in phase $P'$ before starting the execution of Procedure $SymmRV(n,d',\delta)$. Hence, if the agents have not met before, they start the execution of Procedure $SymmRV(n,d',\delta)$ in phase $P$ with the original delay $\delta$. By Lemma \ref{lemma-symm}, they must meet by the end of the execution of Procedure $SymmRV(n,d',\delta)$, hence by the end of phase $P'$.
\end{proof}

Lemma \ref{infeasible} and Theorem \ref{main} imply the following corollary.

\begin{corollary}
1. A STIC $[(u,v),\delta]$ is feasible if and only if either $(u,v)$ are nonsymmetric nodes and $\delta$ is any non-negative integer or $(u,v)$ are symmetric nodes and $\delta \geq Shrink(u,v)$.\\ 
2. Algorithm $UniversalRV$ achieves rendezvous for any feasible STIC with no a priori knowledge. 
\end{corollary}

\section{Complexity of universal anonymous rendezvous}

In this section, we discuss the complexity of universal anonymous rendezvous.  First consider the Procedure $SymmRV(n,d,\delta)$. By Lemma \ref{time-symm}, this procedure takes time $W(n,\delta)\cdot n^d$, where $d$ is the value of $Shrink$ for the initial symmetric pair of nodes, and $W(n,\delta)$ is some polynomial, because  the length $M$ of the UXS $Y(n)$ is polynomial in $n$. Hence it is natural to ask if the dependence on $d$ must be exponential.
We will show that the answer to this question is positive, by constructing a graph in which all pairs of nodes are symmetric, and  any algorithm that achieves rendezvous for all pairs of nodes at initial distance $D$ with delay $\delta=D$ must work in time exponential in $D$. Since the initial distance between symmetric nodes $u$ and $v$ is not smaller than $Shrink(u,v)$, this will show that the dependence on $d$ must also be exponential.  

We start our construction by defining the following tree, see Fig. \ref{Q_h-tree}.

%\begin{definition}%[\textbf{$Q_h$ Tree}]
For any positive integer $h$, the tree $Q_h$ of height $h$ and rooted at node $r$ is defined as follows:
\begin{itemize}
\item All the nodes are unlabeled.  All the leaf nodes are at distance $h$ from the root. Every non-leaf node is of degree 4.
\item Ports at each non-leaf node are labeled $N,S,E$ and $W$, and all edges have either ports $N-S$ or $E-W$ at their extremities.  
\end{itemize}   
%\end{definition}

Next, we modify the tree $Q_h$ by keeping the same set of nodes and adding some edges.
For any positive integer $h$, the graph  $\hat{Q}_h$ is obtained from the tree $Q_h$ by adding some edges between leaves of $Q_h$. There are four types of leaf nodes in $Q_h$. The leaves having the single port number $N, S, E$ and $W$ are called $N$ type, $S$ type, $E$ type and $W$ type leaves, respectively. The tree $Q_h$ has $4\cdot3^{h-1}$ leaves. Among them, there are  $3^{h-1}$ leaves of each type. 
Let $x=3^{h-1}$. For $A \in \{N,S,E,W\}$,  let  $\{A_1,A_2,\dots ,A_x\}$ be the leaves of the $A $ type, ordered in any way. (These labels are only put to clarify the construction of the graph, whose nodes are anonymous).

The set of nodes of $\hat{Q}_h$ is the same as that of $Q_h$ and all the edges from $Q_h$ remain. The set of additional edges between leaves of ${Q}_h$ is defined as follows (see Fig.  \ref{Q_h-tree}).  
\begin{itemize}
\item 
For any $i \leq x$, add the edge joining $N_i$ with $S_i$ and put port $S$ at node $N_i$ and port $N$ at node $S_i$ corresponding to this edge. For any $i \leq x$, add the edge joining $E_i$ with $W_i$ and put port $W$ at node $E_i$ and port $E$ at node $W_i$ corresponding to this edge.
\item 
Add all edges of the cycle $N_1$ --- $S_2$ ---$N_3$ --- ... --- $S_{x-1}$ --- $N_x$ --- $N_1$. For each edge $N_j$ --- $S_{j+1}$, put port $E$ at node $N_j$ and port $W$ at node $S_{j+1}$. For each edge $S_j$ --- $N_{j+1}$, put port $E$ at $S_j$ and port $W$ at $N_{j+1}$. For the edge 
$N_x$ --- $N_1$, put port $E$ at $N_x$ and $W$ at $N_1$.

Add all edges of the cycle $S_1$ --- $N_2$ ---$S_3$ --- ... --- $N_{x-1}$ --- $S_x$ --- $S_1$. For each edge $S_j$ --- $N_{j+1}$, put port $E$ at node $S_j$ and port $W$ at node $N_{j+1}$. For each edge $N_j$ --- $S_{j+1}$, put port $E$ at $N_j$ and port $W$ at $S_{j+1}$. For the edge 
$S_x$ --- $S_1$, put port $E$ at $S_x$ and $W$ at $S_1$.

Add all edges of the cycle $E_1$ --- $W_2$ ---$E_3$ --- ... --- $W_{x-1}$ --- $E_x$ --- $E_1$. For each edge $E_j$ --- $W_{j+1}$, put port $N$ at node $E_j$ and port $S$ at node $W_{j+1}$. For each edge $W_j$ --- $E_{j+1}$, put port $N$ at $W_j$ and port $S$ at $E_{j+1}$. For the edge 
$E_x$ --- $E_1$, put port $N$ at $E_x$ and $S$ at $E_1$.

Add all edges of the cycle $W_1$ --- $E_2$ ---$W_3$ --- ... --- $E_{x-1}$ --- $W_x$ --- $W_1$. For each edge $W_j$ --- $E_{j+1}$, put port $N$ at node $W_j$ and port $S$ at node $E_{j+1}$. For each edge $E_j$ --- $W_{j+1}$, put port $N$ at $E_j$ and port $S$ at $W_{j+1}$. For the edge 
$W_x$ --- $W_1$, put port $N$ at $W_x$ and $S$ at $W_1$.
\end{itemize}

By construction, all nodes of the graph $\hat{Q}_h$ have degree 4 and all of its edges have either ports $N-S$ or $E-W$ at their extremities. These ports could be numbered 0,1,2,3 instead of $N,E,S,W$,
but we prefer the letter notation that visually suggests cardinal directions. 

\begin{figure}[tp]
\centering
\includegraphics[scale=0.5]{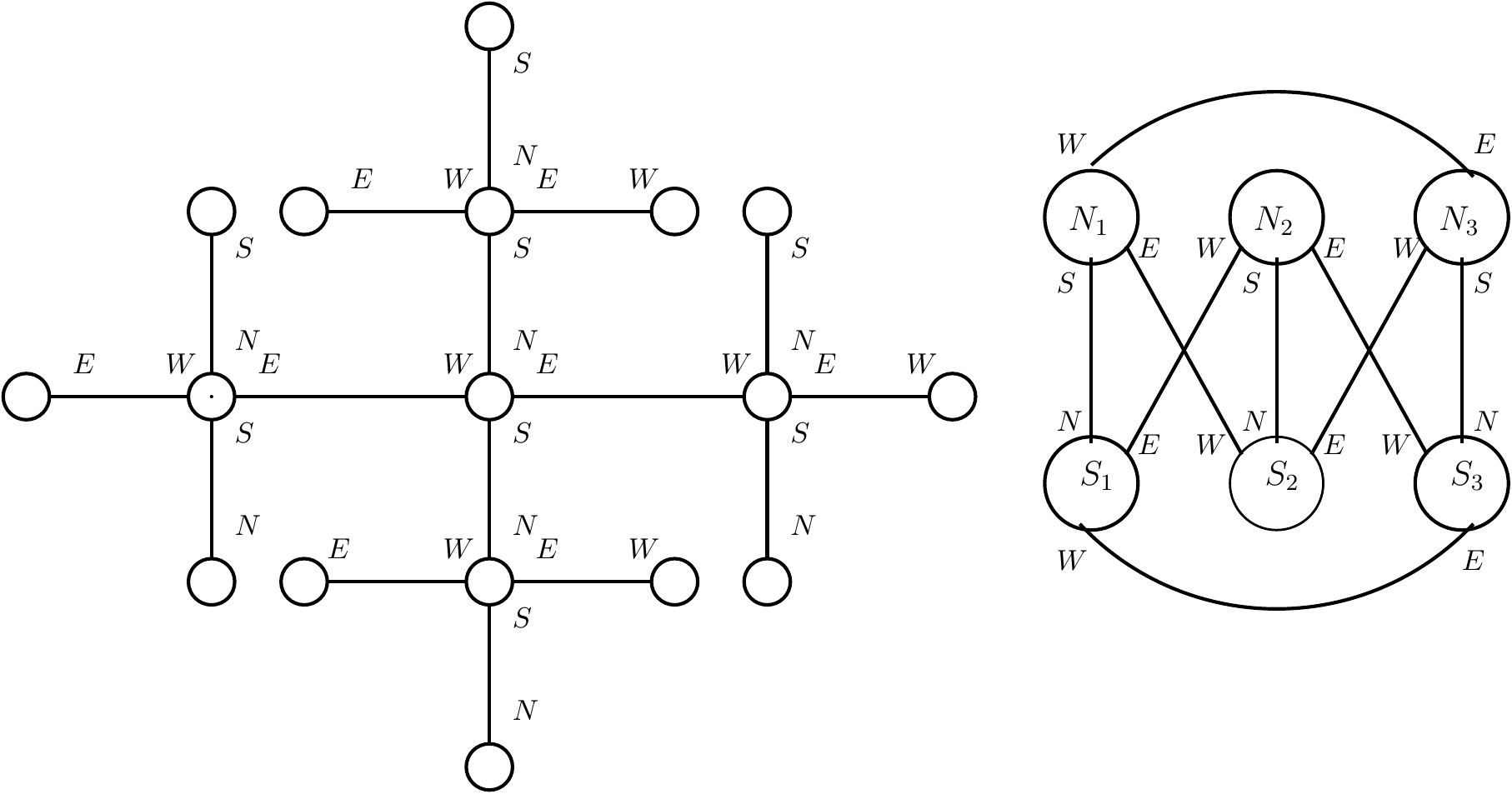}
\caption{The tree $Q_2$ (left) and additional edges between some of its leaves in the graph $\hat{Q_2}$ (right). }
\label{Q_h-tree}
\end{figure}

Fig. \ref{Q_h-tree}. shows the tree $Q_2$ and the additional edges between $N$ type leaves and $S$ type leaves in the 
corresponding graph $\hat{Q}_2$. Notice that the view of each node of $\hat{Q}_h$ is identical, and hence all pairs of nodes are symmetric. For convenience, we will say that the root $r$ of the tree  $Q_h$ is also the root of $\hat{Q}_h$.

Let $D$ be a positive even integer, $D=2k$. Let $h=2^D$ and consider the graph $\hat{Q}_h$.
We define the following set $Z$ of nodes. A node $v$ belongs to $Z$, if there exists a sequence $\gamma =(p_1,\dots, p_k)$ of port numbers, such that  $p_i\in \{N,E\}$ and $v=(\gamma ^{\frown}\gamma)(r)$, where $^{\frown}$ denotes the concatenation. Thus, for example, the node $(NEENEE)(r)$ is in $Z$. By definition, all nodes in $Z$ are at distance $D$ from $r$, and the size of $Z$ is $2^k$. 

The following theorem shows that in the graph $\hat{Q}_h$ the rendezvous time must be exponential in the distance between some initial positions $u$, $v$ of the agents and hence exponential in $Shrink(u,v)$.

\begin{theorem}\label{lb}
Any algorithm that achieves rendezvous for any STIC $[(r,v),D]$ in $\hat{Q}_h$, where $D=2k$, $h=2^D$ and $v \in Z$, must work in time at least $2^{k-1}$.
\end{theorem}

\begin{proof}
%Since in all considered pairs one node $r$ is fixed, we can identify a node $v\in S$ with the sequence $\gamma ^{\frown}\gamma$, such that $v$ is the end of the application $(\gamma ^{\frown}\gamma)(r)$.
Consider a hypothetical algorithm $A$ that achieves rendezvous for any STIC $[(r,v),D]$, where $v \in Z$, in time shorter than $2^{k-1}$.
Since the graph $\hat{Q}_h$ is regular, its nodes do not have labels and all edges have either ports $N-S$ or $E-W$ at their extremities, the agent cannot get any knowledge during the navigation in the graph. Hence any algorithm can be simply coded as a sequence of symbols from the set $\{\epsilon,N,E,S,W\}$ where $\epsilon$ means that the agent stays put in a given step, and $N,E,S,W$ mean that the agent takes the respective port in a given step.  (In other words, there are no conditional statements in such an algorithm).  Since $h=2^D$,  the paths traversed by the agents cannot contain edges added between leaves of $Q_h$ in the construction of $\hat{Q}_h$ because time is shorter than the distance from the initial positions of the agents to the leaves of $Q_h$. Hence, in the rest of the proof we can assume that we work in the tree $Q_h$.

For any path $\pi$ with extremities $a$ and $b$, the {\em reduct} $Rd(\pi)$ is defined as the (unique) simple path from $a$ to $b$. For any $v \in Z$, define $M(v)$ to be $\gamma(r)$, where 
$v=(\gamma ^{\frown}\gamma)(r)$, and define $meet(v)$ to be the node at which a pair of agents starting from $r$ and from $v$ and executing algorithm $A$  accomplishes rendezvous. For any $v\in Z$, consider the paths $\lambda_1(v)$ and $\lambda_2(v)$
to the node $meet(v)$, from $r$ and $v$, respectively, which the agents starting at $r$ and $v$ follow in the execution of algorithm~$A$.

For any node $v\in Z$ there are two possibilities.
Either the reduct $Rd(\lambda_1(v))$ contains the simple path $\pi(v)$ with extremities $r$ and $M(v)$ as a prefix,  or the reduct $Rd(\lambda_2(v))$ contains the simple path $\pi'(v)$ with extremities $v$ and $M(v)$ as a prefix. Hence, for at least one half of all nodes
$v$ in $Z$,  one of these cases must occur. Without loss of generality, suppose that it is the first case. Hence, for at least 
 $y=2^{k-1}$ nodes $v$ in $Z$, call them $v_1,\dots ,v_y$, the agent starting at node $r$ must get to the nodes
 $M(v_1),\dots,M(v_y)$, respectively, in order to meet the other agent at nodes $meet(v_1), \dots ,meet(v_y)$, respectively.
 Let $\rho_i$ be the sequence of ports with terms from $\{N,E\}$, corresponding to the unique simple path from $r$ to $M(v_i)$, for $i=1,\dots ,y$.
 Thus, the initial segment $I$ of  the sequence of ports representing the algorithm $A$, corresponding to the execution of the algorithm that guarantees rendezvous for all initial pairs $(r,v)$, where $v\in Z$,  must contain as subsequences all sequences $\rho_1,\dots ,\rho_y$. Hence the initial segment $I$ must be of length at least  $y=2^{k-1}$, which proves that algorithm $A$ requires time at least $2^{k-1}$.
 This contradiction completes the proof.
\end{proof}

We now turn attention to the complexity of Algorithm $UniversalRV$. Suppose that for some feasible STIC $[(u,v),\delta]$ in a graph of size $n$, the rendezvous is achieved. We want to estimate the time used by Algorithm $UniversalRV$ to accomplish this task.

\begin{proposition}
The time used by Algorithm $UniversalRV$ to accomplish rendezvous for a STIC  $[(u,v),\delta]$ in a graph of size $n$ is 
$O(n+\delta)^{O(n+\delta)}$.
\end{proposition}

\begin{proof}
Consider a STIC $[(u,v),\delta]$ in a graph of size $n$, and let $d=Shrink(u,v)$, in the case when $u$ and $v$ are symmetric. By the design of 
Algorithm $UniversalRV$, rendezvous is accomplished at the latest in phase $P=g(n,d,\delta)$.

%Let $n, d$ and $\delta$  be the real parameters corresponding to this STIC $[(u,v),\delta]$ and $(n,d,\delta) = g^{-1}(P)$ be the real phase of Algorithm $UniversalRV$ in which rendezvous is achieved. So, total $P$ phases are executed of the Algorithm $UniversalRV$. Now, we estimate the time of each phase of the Algorithm $UniversalRV$.

Consider any phase $P'\leq P$, and its corresponding triple $(n',d',\delta') = g^{-1}(P')$. The parameters $n'$, $d'$ and $\delta'$ of phase $P'$ are bounded by  $n+d+\delta$. Since $d<n$, we have $n', d',\delta' \in O(n+\delta)$ (and each of these parameters is actually $\Theta(n+\delta)$, for some phase $P'<P$) . During the execution of phase $P'$ of Algorithm $UniversalRV$, which happens if $d'<n'$, Procedure $AsymmRV(n')$ is executed first and takes time $2(P(n')+\delta')$, hence it uses $A(n,\delta)$ rounds, where $A(n,\delta)$ is some polynomial in $n$ and $\delta$ .

Next, if $d'<n'$ and $d' \leq \delta'$, then Procedure $SymmRV(n',d',\delta')$ is executed. By Lemma \ref{time-symm}, Procedure $SymmRV(n',d',\delta')$ takes time $W(n',\delta')\cdot {(n')}^{d'}$, where $W(n',\delta')$ is some polynomial in $n'$ and $\delta'$. Since
$n', d',\delta' \in O(n+\delta)$, Procedure $SymmRV(n',d',\delta')$ takes time ${O(n+\delta)}^{O(n+\delta)}$. Hence, each phase $P'\leq P$ of Algorithm $UniversalRV$ takes time $O(n+\delta)^{O(n+\delta)}$. 

Now, we count the number of phases $P' \leq P$ of the algorithm $UniversalRV$ that are executed before rendezvous for the
STIC $[(u,v),\delta]$. By the definition of function $f: \mathcal{N}\times \mathcal{N} \longrightarrow \mathcal{N}$, we have
$f(n,d)\in O(n^2+d^2)$. By the definition of function  $g: \mathcal{N}\times \mathcal{N}\times \mathcal{N} \longrightarrow \mathcal{N}$, we have
$g(n,d,\delta) \in O((n^2+d^2)^2+\delta^2)=O(n^4+d^4+\delta^2) \subseteq O(n^4+\delta^2)$. Hence the number of phases of the algorithm $UniversalRV$ to be executed before rendezvous is $O(n^4+\delta^2)$.  
%given by the formula $f(n,d)=n+(n+d-1)(n+d-2)/2$, for any $n$ and $d$, the function $f(n,d)$ is bounded by $O(n^2+d^2)$. The function $g(n,d,\delta)=f(f(n,d),\delta)$ is given by the formula $f(w,\delta)=w+(w+\delta-1)(w+\delta-2)/2$, where $w=f(n,d)$. So, for triple $(n,d,\delta)$, the function $g(n,d,\delta)$ is bounded by $O((n^2+d^2)^2+\delta^2)=O((n+d)^4+\delta^2)$. Since the parameter $d$ will always be smaller than $n$,  the function $g(n,d,\delta)=O(n^4+\delta^2)$. Hence, the number of phases of the algorithm $UniversalRV$ to be executed before rendezvous is $N(n,\delta)=O(n^4+\delta^2)$. 
Thus, the total time used by Algorithm $UniversalRV$ to accomplish rendezvous for the
STIC $[(u,v),\delta]$ is $O(n+\delta)^{O(n+\delta)}$.
\end{proof}

Since Algorithm $UniversalRV$  does not know the parameters $n,d,\delta$ and works in phases, some phases $P'$ preceding the final phase $P$ have the corresponding $d'$ of size $\Theta(n+\delta)$, which forces complexity exponential in $\Theta(n+\delta)$, even if the actual $d$ is much smaller than $n+\delta$. Although we know, by Theorem \ref{lb}, that the complexity of rendezvous has to be exponential in $d$, it is not clear that it has to be exponential in  $\Theta(n+\delta)$. Notice that our lower bound in Theorem \ref{lb} was obtained for a graph in which $n$ was exponential in $d$. Hence it is not precluded that there could exist a universal algorithm with complexity polynomial in $n$ and $\delta$. In fact, a simplified algorithm working only for STICs  $[(u,v),\delta]$ with asymmetric nodes $u,v$, which can be obtained from Algorithm $UniversalRV$
by deleting the Procedure $SymmRV$ in each phase, would indeed be polynomial in $n$ and $\delta$. (This simplified algorithm for asymmetric starting positions would be the version of rendezvous from \cite{CKP} without any knowledge of the size of the graph). This yields the main open problem left by our work.
\begin{quotation}
Does there exist a universal deterministic algorithm which guarantees rendezvous for all feasible STICs in time polynomial in the size of the graph and in the delay between agents?
\end{quotation}

\section{Conclusion}
We characterized all space-time initial configurations for which anonymous deterministic rendezvous is possible, and we designed a universal algorithm that accomplishes rendezvous for all of them without any initial knowledge. Our algorithm works in time exponential in $\Theta(n+\delta)$, where  $n$ is the size of the graph and $\delta$ is the delay between starting times of the agents. While we showed that,
for some initial positions in some graphs, the time must be exponential in the initial distance between the agents, it is not clear if it could not be polynomial in $n$ and $\delta$. 
%Hence  the main open problem left by our work is whether there exists such a universal algorithm which guarantees rendezvous time polynomial in the size of the graph and in the delay.

It should be mentioned that in the asynchronous version of our problem, time cannot be used to break symmetry, as the speed of the agents and the delay between them is controlled by the adversary. Hence in the asynchronous scenario, only space can be used to break symmetry between anonymous agents, and this was completely solved in \cite{GP}. On the other hand, the synchronous randomized counterpart of our problem is straightforward, and follows from the fact that two random walks meet with high probability in time polynomial in the size of the graph \cite{MU}.

%\newpage

\end{document}